\documentclass[journal]{IEEEtran}
\usepackage{cite}
\ifCLASSINFOpdf
\else
\fi

\usepackage{cite}
\usepackage[cmex10]{amsmath}
\usepackage{amsthm}
\usepackage{array}
\usepackage{xcolor}
\usepackage{graphicx}
\usepackage{euscript}
\usepackage{amsfonts}
\usepackage{bigstrut}
\usepackage{tabularx}
\usepackage{cite}
\usepackage{epstopdf}
\onecolumn
\usepackage{float}
\usepackage{dsfont}
\usepackage{hyperref}

\usepackage{widetext} 

\usepackage{url}

\newtheorem{Theorem}{Theorem}

\newtheorem{Definition}{Definition}

\begin{document}

\title{Nearest Neighbour Distance Distribution in Hard-Core Point Processes}
\author{Akram~Al-Hourani \IEEEmembership{Member,~IEEE}, Robin J. Evans~\IEEEmembership{Life Fellow,~IEEE}, and \\Sithamparanathan Kandeepan \IEEEmembership{Senior Member,~IEEE}
\thanks{Manuscript received xx-xxx-2016; revised xx-xxx-2016. The research was supported by the Australian Research Council Discovery Project Grant “Cognitive Radars for Automobiles”,	ARC grant number DP150104473.}
\thanks{A. Al-Hourani, R. Evans are with the department of electrical and electronic engineering, The University of Melbourne, Melbourne, Australia. E-mail: \{akram.hourani, robinje\}@unimelb.edu.au}
\thanks{S. Kandeepan, is with the school of electrical and computer engineering, RMIT University, Melbourne, Australia. E-mail: kandeepan.sithamparanathan@rmit.edu.au}
}
\maketitle

\begin{abstract}
In this paper we present an analytic framework for formulating the statistical distribution of the nearest neighbour distance in hard-core point processes. We apply this framework to Mat\'{e}rn hard-core point process (MHC) to derive the cumulative distribution function of the contact distance in three cases. The first case is between a point in an MHC process and its nearest neighbour from the same process. The second case is between a point in an independent Poisson point process and the nearest neighbour from an MHC process. The third case is between a point in the complement of an MHC process and its sibling MHC process. We test the analytic results against Monte-Carlo simulations to verify their consistency.
\end{abstract}


\IEEEpeerreviewmaketitle

\section{Introduction}
\IEEEPARstart{H}{ard}-core point processes have wide applications in modelling the geometrical randomness in wireless network nodes \cite{Haenggi_Interference} and in many other fields \cite{teichmann2013generalizations}. Their ability to capture the mutual repulsion between points can better represent many physical systems. For example, the deployment of cellular wireless base stations (BS) has a certain degree of randomness because of the uncertainties related to the availability of land, utility, planning permits and other socio-economic factors, however the locations of two BSs have strong correlation and cannot occur closer than a predefined minimum distance. Another example is modelling vehicle locations on a road, where two cars cannot drive in very close proximity. These constraints in physical systems are usually ignored when modelling using Poisson point process (PPP), with zero correlation between the location of points. However its tractability might sometimes justify the resulting inaccuracy.

An alternative to PPP, hard-core point processes which impose a protection ball $b(x,\delta)$ around every point $x \in \Phi$ in the process $\Phi$, where no other point is allowed to be within this ball of radius $\delta$ called the \emph{hard-core distance}. However, hard-core processes are less tractable and more difficult to analyse. In particular, the nearest neighbour distance for PPP has long been known \cite{Book_Haenggi} to follow a generalized gamma distribution, however it is not known in many of the hard-core processes. One important hard-core point process is Mat\'{e}rn hard-core point process (MHC) \cite{Matern_Original}\cite{Book_Haenggi} that has a tractable density and it efficiently \emph{thins} a parent PPP to satisfy the hard-core condition, where this thinning transformation naturally reduces the density of points. MHC has significant applications in modelling carrier sense multiple access (CSMA), where the sensing radius is modelled as the hard-core distance, in this case the contact distance is of a paramount importance to determine the closest access point \cite{6365639} or the closest cluster-head in wireless sensor networks \cite{7417659}. MHC can also be applied to cellular base stations, where the received signal power is mainly dependent on the contact distance \cite{7460095}.

In this paper we present an analytic framework to calculate the cumulative distribution function of the contact distance to the first nearest neighbour. This framework can be applied to any point process given that its conditional retention probability function $\eta_{m \to c}(r)$ is known, and given that it is a result of a thinning transformation from a parent homogeneous PPP. We employ this approach on the MHC process to obtain the contact distance for three different cases. Between a point in:
\begin{itemize}
	\item MHC process and its nearest neighbour from the same MHC process.
	\item PPP process and the nearest neighbour from an independent MHC process.
	\item CMHC process and its nearest neighbour from the sibling MHC process. 
\end{itemize}
Where CMHC refers to the complementary Mat\'{e}rn hard-core point process defined as the residual points of the thinning transformation from PPP to MHC.

\section{Analytic Framework}
A hard-core point process can be generated from a thinning transformation of a parent PPP process $\Phi_p$, where some points migrate to the child hard-core process $\Phi_c$. Given the following two conditions: (i) two points $x_o \in \Phi_m$ and $x \in \Phi_p$ exist and are separated by a distance $r$, where $\Phi_m$ is a certain point process with respect to which the nearest distance to $\Phi_c$ is measured, and (ii) the ball centred at $x_o$ of radius $r$ is void of points in $\Phi_c$, we define the conditional thinning Palm-probability as the probability of the point $x$ to migrate to the child process $\Phi_c$.
\begin{Definition}
	The conditional thinning Palm-probability:
	\begin{equation}
	 \eta_{m \to c}(r)\stackrel{\bigtriangleup}{=} \mathbb{P}\left[x \in \Phi_c  | \Phi_c \cap b(x_o,r) = \emptyset, x_o \in \Phi_m \right],
	\end{equation}
\end{Definition} 

Accordingly, the following theorem is shown to hold:
\begin{Theorem}
	The cumulative distribution function of the contact distance is given by:
	\begin{equation} \label{eq_main}
		F_{{m \to c}}(R) = 1 - \exp \left( - \int_{0}^{R} 2 \pi r \lambda_p \eta_{m\to c}(r) ~ \mathrm{d} r\right)
	\end{equation}
\end{Theorem}

\begin{proof}
Starting from a parent PPP $\Phi_p$ with a homogeneous intensity $\lambda_p$ and given a bounded Borel set $A \subset \mathbb{R}^2$, then the void probability is given by \cite{Book_Baddeley} as $\exp(-|A|\lambda_p)$ representing the probability that no point exists in $A$, while $|.|$ is the Lebesgue measure. Thus, if we take an infinitesimal annulus of width $\mathrm{d}r$ and radius $r$, then its void probability is $v=\exp(-2\pi r \lambda_p~\mathrm{d}r)$. 

Now assume a hard-core thinning process on $\Phi_p$ having a conditional thinning probability $\eta_{m \to c}(r)$. Then the probability of having a void infinitesimal annulus $\mathcal{A}(x_o,r,\mathrm{d}r)$ in the child hard-core process given two conditions: (i) a void ball $b(x_o,r)$ and (ii) a point $x_o \in \Phi_m$, is obtained as follows:

\begin{align}\label{Eq_3}
v_{m \to c}(r) &= \mathbb{P} \left[\Phi_c \cap \mathcal{A}(x_o,r,\mathrm{d}r) = \emptyset |   \Phi_c \cap b(x_o,r) = \emptyset , x_o \in \Phi_m \right] \nonumber \\
&=1 - \eta_{ m \to c}(r) \left[1- \exp\left(-2\pi r \lambda_p  \mathrm{d}r\right) \right] \nonumber \\
&\approx 1- 2 \pi r \lambda_p~\eta_{m \to c}(r) ~  \mathrm{d}r, 
\end{align} 
by using Taylor expansion around zero, where $e^x=1+x+O(x^2)$, so that $e^x\approx1+x$  as $x\to 0$.

\begin{figure}
	\centering
	\includegraphics[width=0.75\linewidth]{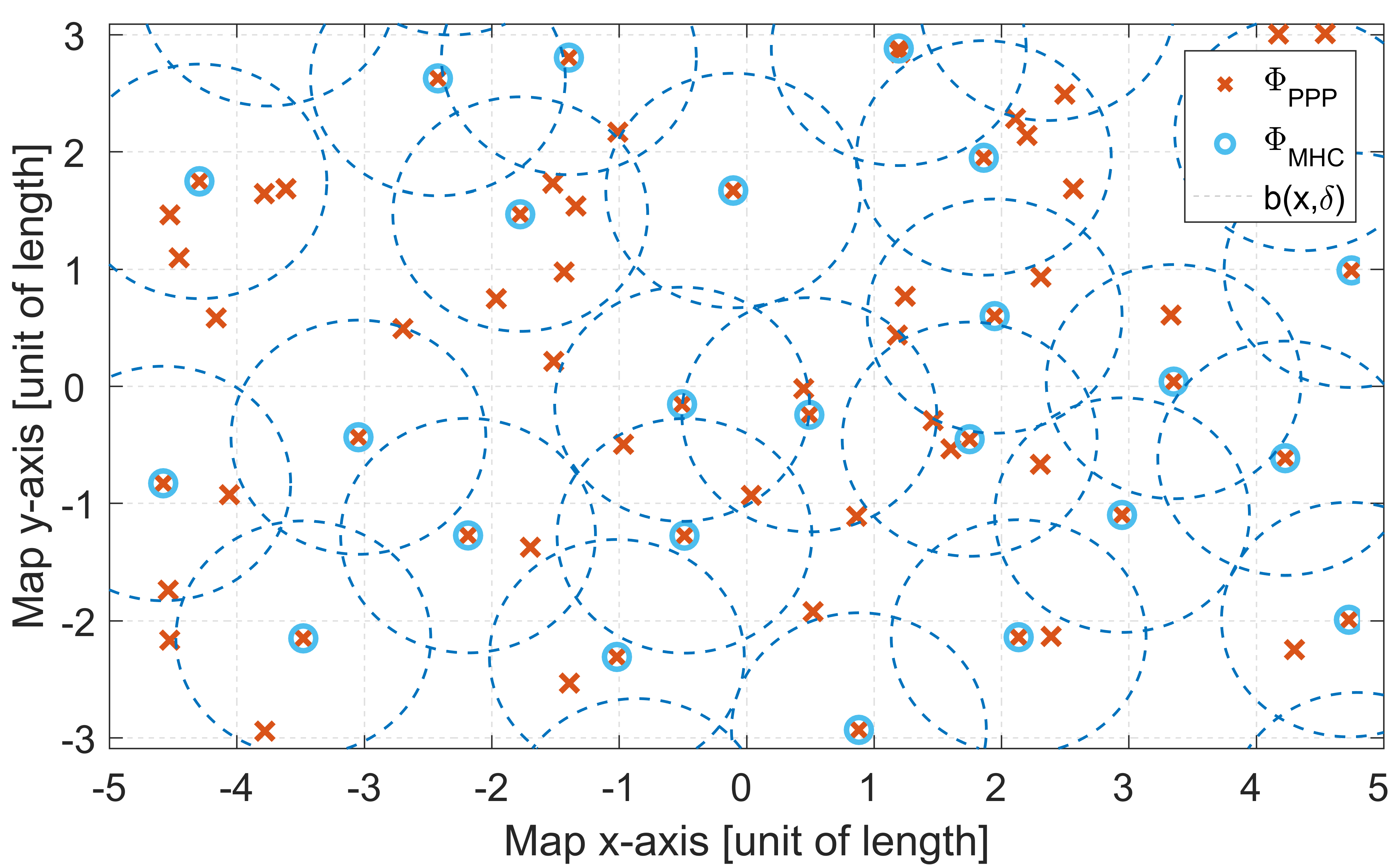}
	\caption{The hard-core region indicated in dotted line for every member $x\in \Phi_{\mathrm{MHC}}$ thinned out of a homogeneous parent PPP $\Phi_\mathrm{PPP}$.}\label{Fig_MHC}
\end{figure}

Let $A_o, A_1, \dots A_n,\dots A_{N-1}$ represent the events of all annuli between the point $x_o$ and circle $b(x_o,R)$ to be void, these annuli are $\mathcal{A}(x_o,r_o,r_o+\Delta r),~\mathcal{A}(x_o,r_1,r_1+\Delta r),\dots,~\mathcal{A}(x_o,r_n,r_n+\Delta r),\dots,~\mathcal{A}(x_o,r_{N-1},r_{N-1}+\Delta r)$, where $r_n = n \Delta r$, $\Delta r=\frac{R}{N}$, and $N$ is a positive integer representing the number of annuli. Then the probability of a void ball $b(x_o,R)$ in $\Phi_c$ given a point $x_o \in \Phi_m$  is:
\begin{align}\label{Eq_4}
	&V(R) = \mathbb{P} \left[\Phi_c \cap b(x_o,R) = \emptyset | x_o \in \Phi_m \right] \nonumber \\
	&= \mathbb{P} \left[A_o  \cap \dots \cap A_{N-1} | x_o \in \Phi_m\right] \nonumber \\
	&= \mathbb{P} \left[A_o | x_o \in \Phi_m \cap \dots \cap A_{N-1} | x_o \in \Phi_m \right] \nonumber \\
	&\stackrel{(a)}{=} \mathbb{P} \left[\bigcap_{n=1}^{N-1} B_n \right] \stackrel{(b)}{=} \prod_{n=2}^{N-1} \mathbb{P} \left[B_n|\bigcap_{i=1}^{n-1}B_i\right]\times \mathbb{P} \left[B_o\right]\nonumber \\
	&= \lim\limits_{N\to\infty} \prod_{n=1}^{N-1} v_{m \to c}(r_n) = \lim\limits_{N\to\infty}\prod_{n=1}^{N-1} \left[ 1- 2 \pi r \lambda_p~\eta_{m \to c}(r_n) ~  \Delta r \right]
\end{align}
where in step $(a)$ we took the event $B_n  \stackrel{\bigtriangleup}{=} A_n | x_o \in \Phi_m$. Step $(b)$ follows form the multiplication rule of dependent events, and the last step by substituting from (\ref{Eq_3}). By taking the logarithm of both sides of (\ref{Eq_4}) we get: 
\begingroup
\allowdisplaybreaks
\begin{align}
&\ln\left[V(R)\right] = \lim\limits_{N\to\infty} \sum_{n=1}^{N-1}\ln \left[ 1- 2 \pi r \lambda_p~\eta_{m \to c}(r_n) ~  \Delta r \right]\nonumber \\
&\stackrel{(a)}{\approx} - \lim\limits_{N\to\infty} \sum_{n=1}^{N-1} 2 \pi r \lambda_p~\eta_{m \to c}(r_n) ~  \Delta r \stackrel{(b)}{=} - \int_{r=0}^{R} 2 \pi r \lambda_p~\eta_{m \to c}(r) ~  \mathrm{d}r,
\end{align}
\endgroup
step $(a)$ follows from Taylor expansion  of the natural logarithm around zero, where $\ln(1+x)=x+O(x^2)$, thus $\ln(1+x)\approx x$ as $x \to 0$. Using the theory of Riemann\footnote{The proof is detailed in the Appendix} integration it can be shown that step $(b)$ converges to Riemann integral as the number of partitions (annuli) $N\to \infty$. Accordingly $V(R)$ is found by taking the natural exponential of both sides. While the event of having a neighbour $x \in \Phi_c$ to a point $x_o \in \Phi_m$ within a radius $R$ is the complement of the void probability of $b(x_o,R)$, thus: $F_{m\to c}(R)=1-V(R)$ and the proof of Theorem 1 is complete.
\end{proof}
We can test this theorem against the trivial thinning from PPP to the same PPP, ie $\Phi_p \stackrel{t}{\to}  \Phi_c=\Phi_p$, and taking $\Phi_m=\Phi_p$, thus $\eta_{m \to c}(r)= 1$ (because there is no correlation between points in PPP), and the contact distance CDF reduces to the well-known expression of PPP:
\begin{align} \label{eq_6}
		F_{{PPP \to PPP}}(R) &= 1 - \exp \left( - \int_{0}^{R} 2 \pi r\lambda_p ~ \mathrm{d} r\right) \nonumber \\
		&= 1- \exp\left( - \pi \lambda_p R^2\right)
\end{align}

\section{MHC to MHC Contact Distance}
In this study we focus on MHC type-II since it has a higher density than MHC type-I \cite{chiu2013stochastic}\cite{Haenggi_Interference}, thus yielding a more efficient child point process out of the generating parent PPP process. To construct an MHC type-II process,  we start with some real number $\delta>0$ and a parent stationary Poisson point process (PPP) denoted by $\Phi_p$, having a homogeneous intensity of $\lambda_p$ (points per unit area), and then for each point $x \in \Phi_p$ a uniformly distributed random mark $\mathcal{M}(x) = \boldsymbol{t} ~\sim \mathcal{U}(0,1)$ is assigned. A point $x \in \Phi_p$ will be flagged-for-removal if its mark is not the least compared to all other points within a ball of radius $\delta$, $b(x,\delta)$, around it.  Only after this test is done for all points in $\Phi_p$, the points that were flagged-for-removal are removed. The retained points, constituting the resulting  MHC process described by \cite{Book_Haenggi}:
\begin{align}\label{Eq_Matern}
\Phi_{\mathrm{MHC}} \stackrel{\bigtriangleup}{=}  \{ x \in \Phi_p : &~\mathcal{M} ( x ) < \mathcal{M} ( y ) \nonumber \\
&\forall y \in \Phi_p \cap b(x,\delta)  \setminus \{ x \}  \},
\end{align}
where $\delta$ is the \emph{hard-core distance}. In Fig.\ref{Fig_MHC} we illustrate a parent PPP thinned to an MHC process with parameters $\delta=1$ and $\lambda_p=1$. 
It can be easily shown that the resulting MHC has an intensity of \cite{Book_Haenggi}:
\begin{equation}\label{Eq_MHC_Density}
\lambda_{\mathrm{MHC}} = \frac{1-\exp(-\lambda_p \pi \delta^2)}{\pi \delta^2}.
\end{equation}

\begin{figure*}[!t]
	\normalsize
	\begin{align} \label{eq_9}
	&\kappa_1(r,\delta) = \int_{0}^{1} \int_{0}^{t_o} \exp(-  t_o \lambda_p \pi \delta^2 ) \times \exp \left[- t \lambda_p (\pi \delta^2- l_2(r,\delta)) \right] \mathrm{d}t \mathrm{d}t_o \nonumber \\
	&\kern 3em+\int_{0}^{1} \int_{0}^{t} \exp\left[-  t \lambda_p (\pi \delta^2 -\left[l_2(r,\delta)-l_1(r,\delta)\right] ) \right] \times \exp \left[- t_o \lambda_p (\pi \delta^2- l_1(r,\delta)) \right] \mathrm{d}t_o \mathrm{d}t \nonumber \\
	&= \frac{1-e^{-\pi  \delta ^2 \lambda_p }}{\pi ^2 \delta ^4 \lambda_p ^2-\pi  \delta ^2 \lambda_p ^2 l_2(r,\delta)}+\frac{e^{\lambda_p  \left(l_2(r,\delta)-2 \pi  \delta ^2\right)}-1}{\lambda_p ^2 \left(-3 \pi  \delta ^2 l_2(r,\delta)+l_2(r,\delta)^2+2 \pi ^2 \delta ^4\right)} + \frac{\frac{1-e^{-\lambda_p  \left(l_1(r,\delta)-l_2(r,\delta)+\pi  \delta ^2\right)}}{l_1(r,\delta)-l_2(r,\delta)+\pi  \delta ^2}+\frac{e^{\lambda_p  \left(l_2(r,\delta)-2 \pi  \delta ^2\right)}-1}{2 \pi  \delta ^2-l_2(r,\delta)}}{\lambda_p ^2 \left(\pi  \delta ^2-l_1(r,\delta)\right)} 
	\end{align} 
\end{figure*}
\begin{figure}
	\centering
	\includegraphics[width=0.7\linewidth]{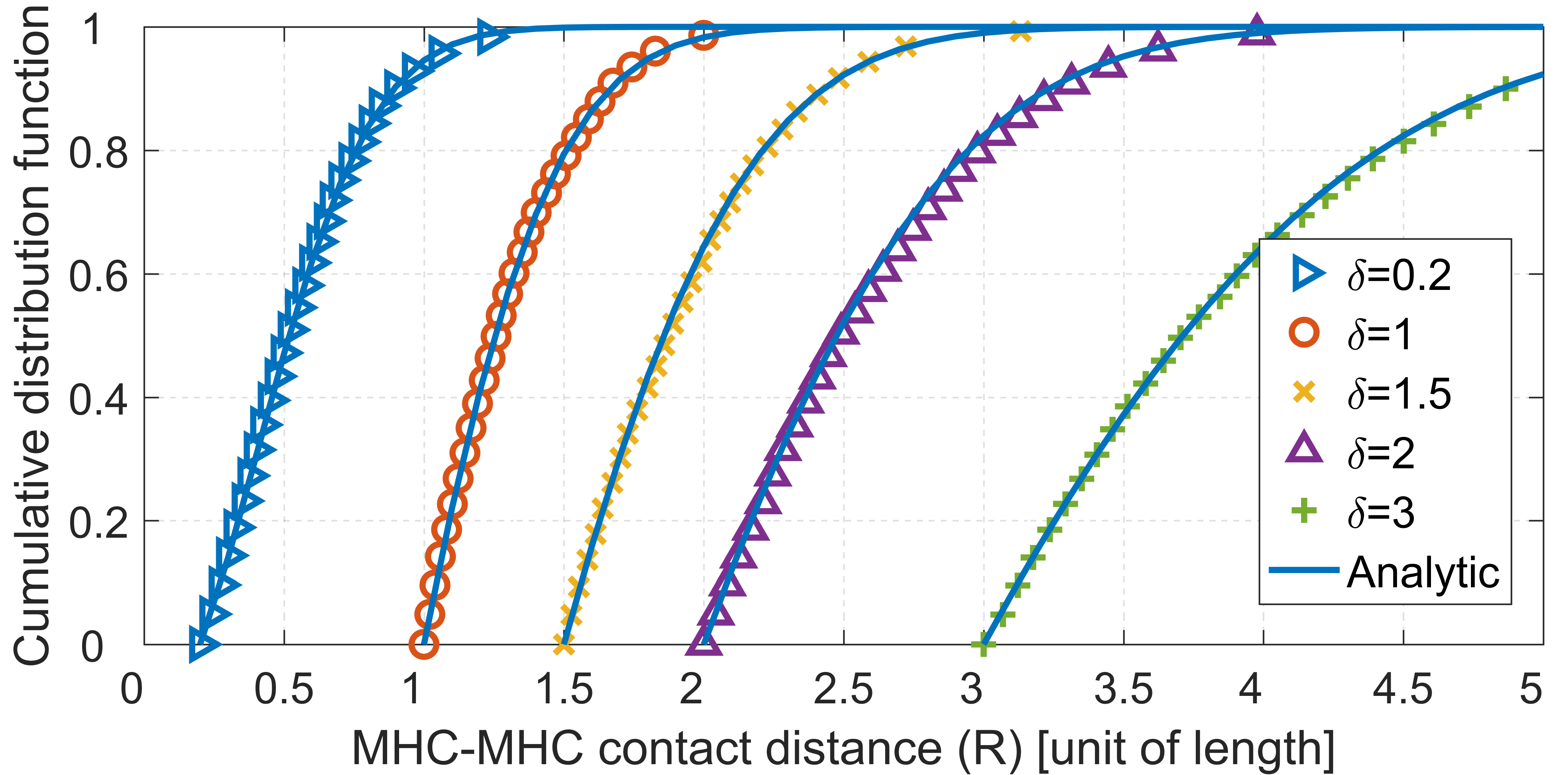}
	\includegraphics[width=0.7\linewidth]{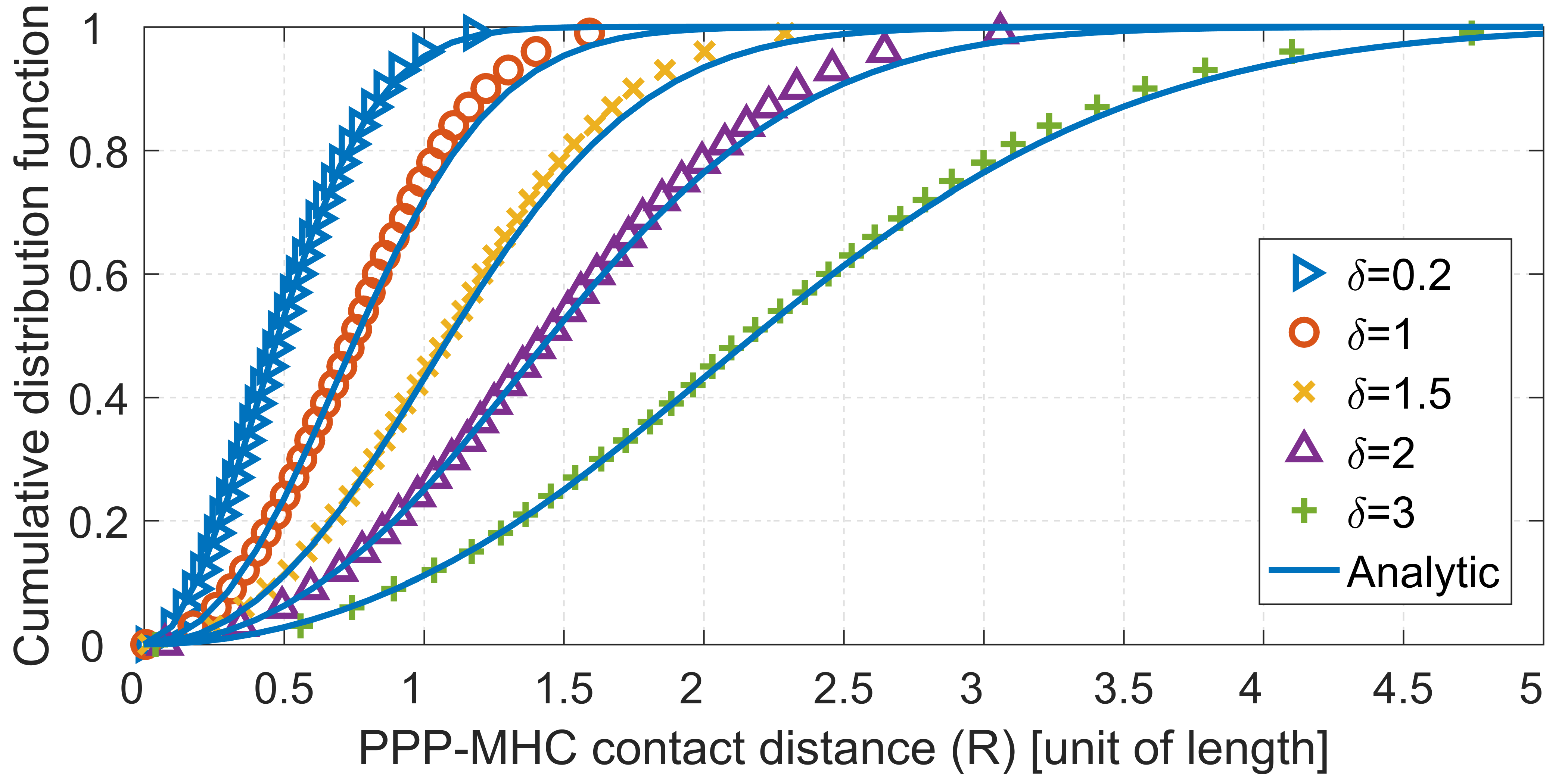}
	\caption{The cumulative distribution function of the distance in two cases (up to down): (i) between an MHC point and its nearest neighbour from the same process, (ii) a point in PPP and its nearest neighbour in an independent MHC process, comparing simulations (discrete points) and the analytical plot (continuous line).
	}\label{Fig_MHC_MHC}
\end{figure}
\begin{figure}[ht]
	\centering
	\includegraphics[width=0.5\linewidth]{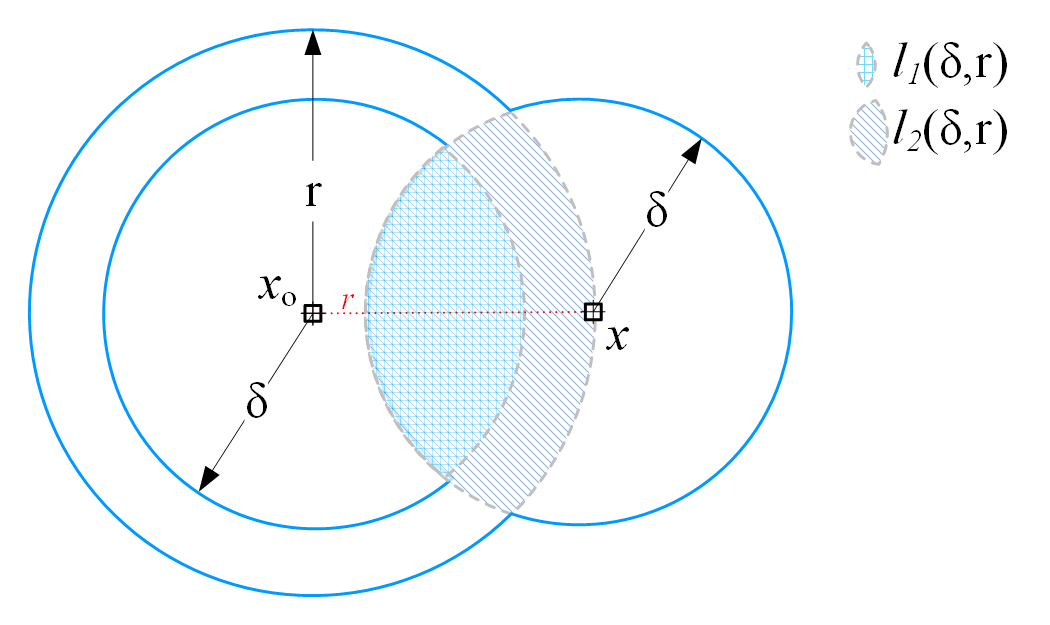}
	\includegraphics[width=0.5\linewidth]{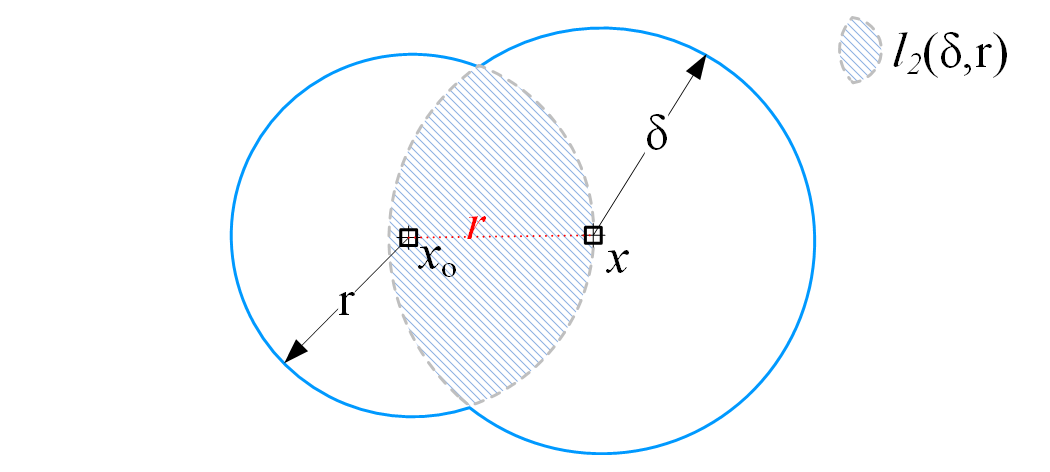}
	\caption{(Up) An illustration of the two-point interaction between $x_o$ and $x$ given a void ball, showing the resulting lenses used in the integration (\ref{eq_9}). (Down) the interaction in the case of PPP $\to$ MHC in (\ref{eq_15}).}\label{Fig_lenses}
\end{figure}

Our target is to calculate the CDF of the contact distance from a typical point in MHC to its nearest point in the same process, thus we need the conditional thinning probability $\eta_{\mathrm{MHC}\to \mathrm{MHC}}$, formed as the following:
\begingroup
\allowdisplaybreaks
\begin{align}\label{eq_8}
&\eta_{\mathrm{MHC}\to \mathrm{MHC}}(r,\delta)=\mathbb{P} \left[x \in \Phi_\mathrm{MHC}  | \right.\nonumber \\
&\quad \quad \quad \quad \left. \Phi_\mathrm{MHC} \cap \mathcal{A}(x_o,\delta,r) = \emptyset, x_o \in \Phi_\mathrm{MHC} \right] \nonumber \\
&=\frac{\mathbb{P}\left[x \in \Phi_\mathrm{MHC} \cap x_o \in \Phi_\mathrm{MHC} | \Phi_\mathrm{MHC} \cap \mathcal{A}(x_o,\delta,r) = \emptyset \right]}{\mathbb{P} \left[x_o \in \Phi_\mathrm{MHC} | \Phi_\mathrm{MHC} \cap \mathcal{A}(x_o,\delta,r)=\emptyset\right]} \nonumber \\
&\stackrel{(a)}{\approx}\frac{\mathbb{P}\left[x \in \Phi_\mathrm{MHC} \cap x_o \in \Phi_\mathrm{MHC} | \Phi_p \cap \mathcal{A}(x_o,\delta,r) = \emptyset \right]}{\mathbb{P} \left[x_o \in \Phi_\mathrm{MHC} | \Phi_p \cap \mathcal{A}(x_o,\delta,r)=\emptyset\right]} \nonumber \\
&=\frac{\kappa_1(r,\delta)}{\kappa_2(r,\delta)},
\end{align}
\endgroup
where we utilized $\mathcal{A}(x_o,\delta,r)$ instead of $b(x_0,r)$ since the ball $b(x_0,\delta)$ is implicitly given to be empty when a point $x_o \in \Phi_\mathrm{MHC}$, the approximation in $(a)$ holds very tight since the conditional thinning probability of an MHC point given a CMHC point is extremely week as verified by simulation\footnote{A void region of MHC points can still have CMHC points, but as the influence of CMHC on MHC is very week, we can consider the original parent process $\Phi_p$ as void when studying the influence on MHC points.}. The numerator $\kappa_1(r,\delta)$ is the two-point Palm probability \cite{chiu2013stochastic} but conditioned on having a void annulus $\mathcal{A}(x_o,\delta,r)$ in $\Phi_p$, thus can be calculated as the following; assume two points $x_o \in \Phi_p$ and $x \in \Phi_p$ exist in the parent PPP with marks $\mathcal{M}(x_o) = t_o$ and $\mathcal{M}(x)=t$ respectively, then we have two mutually exclusive events; either $t_o \ge t$ or $t > t_o$. Accordingly, $\kappa_1(r,\delta)$ is formulated in (\ref{eq_9}) for $r>\delta$, where $l_1(r,\delta)$ is the symmetrical lens formed by the intersection of $b(x_o,\delta)$ and $b(x,\delta)$, while $l_2(r,\delta)$ is the asymmetrical lens formed by the intersection of $b(x_o,r)$ and $b(x,\delta)$, given by: 
\begin{equation}
l_1(r,\delta)=     \left\{
\begin{array}{ll}
2 \delta ^2 \cos^{-1} \left(\frac{r}{2 \delta }\right) -\frac{1}{2} r \sqrt{4 \delta ^2-r^2}&~:~  0<r\leq 2 \delta  \\
0 &~:~ r>2 \delta,  \\
\end{array}
\right.
\end{equation}
\begin{align}
	l_2(r,\delta) =  \left\{
\begin{array}{ll}
\pi r^2 &: 0<r< \frac{\delta}{2} \\
r^2 \cos^{-1} \left(1-\frac{\delta ^2}{2 r^2}\right)+\delta ^2 \cos^{-1}\left(\frac{\delta}{2 r}\right) \nonumber \\ \quad\quad\quad\quad\quad \quad\quad\quad-\frac{1}{2} \delta  \sqrt{4 r^2-\delta ^2} &: r \ge \frac{\delta}{2} 
\end{array}
\right. 
\end{align}
These lenses are shown in Fig. \ref{Fig_lenses}. The formulation of $\kappa_1(r,\delta)$ in (\ref{eq_9}) follows because the density of points having a mark less than $t$ is equal to $\lambda_p t$.

The denominator is calculated as the following:
\begingroup
\allowdisplaybreaks
\begin{align}
\kappa_2(r,\delta) &= \mathbb{P} \left[x_o \in \Phi_\mathrm{MHC} | \Phi_p \cap \mathcal{A}(x_o,\delta,r)=\emptyset\right] \nonumber \\
&\stackrel{(a)}{=} \mathbb{P} \left[x_o \in \Phi_\mathrm{MHC} \right] = \int_{0}^{1} \exp(-t_o \lambda_p \pi\delta^2) \mathrm{d}t_o\nonumber \\ 
& = \frac{1-\exp(- \lambda_p \pi\delta^2)}{\lambda_p \pi\delta^2} \stackrel{(b)}{=} \rho_\mathrm{MHC}
\end{align}
\endgroup
where $(a)$ follows from the fact that the probability is independent of the condition, while $(b)$ is known in the literature as the MHC retention probability \cite{chiu2013stochastic}. Thus the contact distance CDF is obtained by applying the result of Theorem 1:
\begin{equation}\label{eq_12}
	F_{{\mathrm{MHC} \to \mathrm{MHC}}}(R|\delta) = 1 - \exp \left( - \int_{\delta}^{R} 2 \pi \lambda_p \frac{\kappa_1(r,\delta)}{\rho_\mathrm{MHC}} ~ r\mathrm{d} r\right),
\end{equation}
conditioned on a given hard-core distance $\delta$, where this CDF is plotted in Fig. \ref{Fig_MHC_MHC} for fixed parent intensity $\lambda_p=1~[\text{unit length}]^{-2}$, and different $\delta$. Noting that the lower integration limit in (\ref{eq_12}) starts from $\delta$ because of the hard-core constraint on the minimum allowable distance, where $\kappa_1(r,\delta)=0$ for $r \le \delta$.

\section{PPP to MHC Contact Distance}
In this section we study the contact distance between a typical point in a PPP process and its nearest neighbour from an independently generated MHC process. The rationale behind this can be found in wireless communications for example when users are distributed according to an independent homogeneous PPP process and the base stations as MHC. Thus we can write:
\begin{align}\label{eq_14}
 &\eta_{\mathrm{PPP}\to \mathrm{MHC}}=\mathbb{P} \left[x \in \Phi_\mathrm{MHC}  | \right.\nonumber \\
 &\quad \quad \quad \quad \left. \Phi_\mathrm{MHC} \cap b(x_o,r) = \emptyset, x_o \in \Phi_\mathrm{PPP} \right] \nonumber \\
 &\stackrel{(a)}{=}\mathbb{P} \left[x \in \Phi_\mathrm{MHC}  | \Phi_\mathrm{MHC} \cap b(x_o,r) = \emptyset \right] \nonumber \\
 &~\approx\mathbb{P} \left[x \in \Phi_\mathrm{MHC}  | \Phi_p \cap b(x_o,r) = \emptyset \right],
\end{align}
step $(a)$ is because of $\Phi_m=\Phi_\mathrm{PPP}$ and $\Phi_c=\Phi_\mathrm{MHC}$ are independent by assumption, while the last step is taken similar to step $(a)$ in (\ref{eq_8}). Accordingly we can write:
\begin{align}\label{eq_15}
\eta_{\mathrm{PPP}\to \mathrm{MHC}}(r,\delta)&=\int_{0}^{1} \exp\left[- t \lambda_p \left(\pi \delta^2 -  l_2(r,\delta) \right) \right] \mathrm{d}t\nonumber \\
&=\frac{1-\exp\left[- \lambda_p \left(\pi \delta^2 -  l_2(r,\delta) \right)\right]}{\lambda_p \left(\pi \delta^2 -  l_2(r,\delta)\right)}.
\end{align}
Thus, using Theorem 1, the CDF is calculated from:
\begin{equation}
	F_{{\mathrm{PPP} \to \mathrm{MHC}}}(R|\delta) = 1 - \exp \left( - \int_{0}^{R} 2 \pi \lambda_p r \eta_{\mathrm{PPP}\to \mathrm{MHC}}(r,\delta) ~ \mathrm{d} r\right),
\end{equation}
where the plot is shown in Fig. \ref{Fig_MHC_MHC} for different $\delta$ and for $\lambda_p=1~[\text{unit length}]^{-2}$.
The contact distance can be tested against the case when $\delta \to 0$ and accordingly MHC converges to PPP $(\Phi_c=\Phi_\mathrm{MHC} )\stackrel{d}{\to}(\Phi_p = \Phi_\mathrm{PPP})$, where $d$ refers to the \emph{distribution}, and the conditional probability asymptote is:
\begin{equation}
\lim\limits_{\delta \to 0} \eta_{\mathrm{PPP}\to \mathrm{MHC}}(r,\delta) =1 ,~\forall r>0,
\end{equation}
similar to (\ref{Eq_3}), we can reduce (\ref{eq_main}) to the well-known expression of PPP, however we should note here that in the latter case we are measuring the contact distance between a $\Phi_m=\Phi_\mathrm{PPP}$ and another PPP process $\Phi_c=\Phi'_\mathrm{PPP}$, this is identical to measuring the contact distance within $\Phi_c=\Phi'_\mathrm{PPP}$ explained by Palm theory \cite{Book_Baddeley}.

\section{CMHC to MHC Contact Distance Approximation}
When an MHC hard-core thinning is performed on a parent PPP process a resulting child MHC is formed, while the residual points form another process named the complementary-MHC or CMHC defined as:
\begin{equation}\label{Eq_CMatern}
\Phi_\mathrm{CMHC}\stackrel{\bigtriangleup}{=}  \{ x \in \Phi_p : \mathcal{M} ( x ) \geq \mathcal{M} ( y ) \exists y \in \Phi_p \cap b(x,\delta)  \setminus \{ x \}  \},
\end{equation}
where $\Phi_\mathrm{MHC} \cup \Phi_\mathrm{CMHC} = \Phi_p$ and $\Phi_\mathrm{MHC} \cap \Phi_\mathrm{CMHC} = \emptyset$. The approximated  conditional thinning probability is formed as the following:
\begin{equation}
\eta_{\mathrm{CMHC}\to\mathrm{MHC}}(r,\delta) = \mathbb{P} \left[x \in \Phi_\mathrm{MHC}  |  \Phi_\mathrm{MHC} \cap b(x_o,r) = \emptyset, x_o \in \Phi_\mathrm{CMHC} \right],
\end{equation}
however, there is no known method to obtain this probability in a closed form. Instead we relay on the insignificant correlation between the CMHC process and its sibling MHC, as verified by simulation. Accordingly we treat $x_o \in \Phi_\mathrm{CMHC}$ as if it is originated from an independent PPP $x_o \in \Phi_\mathrm{PPP}$. Thus, we consider $\eta_{\mathrm{CMHC}\to\mathrm{MHC}}\approx\eta_{\mathrm{CMHC}\to\mathrm{PPP}} $. This approximation holds well as indicated in Fig. \ref{Fig_CMHC_MHC}.

\begin{figure}[t]
	\centering
	\includegraphics[width=0.75\linewidth]{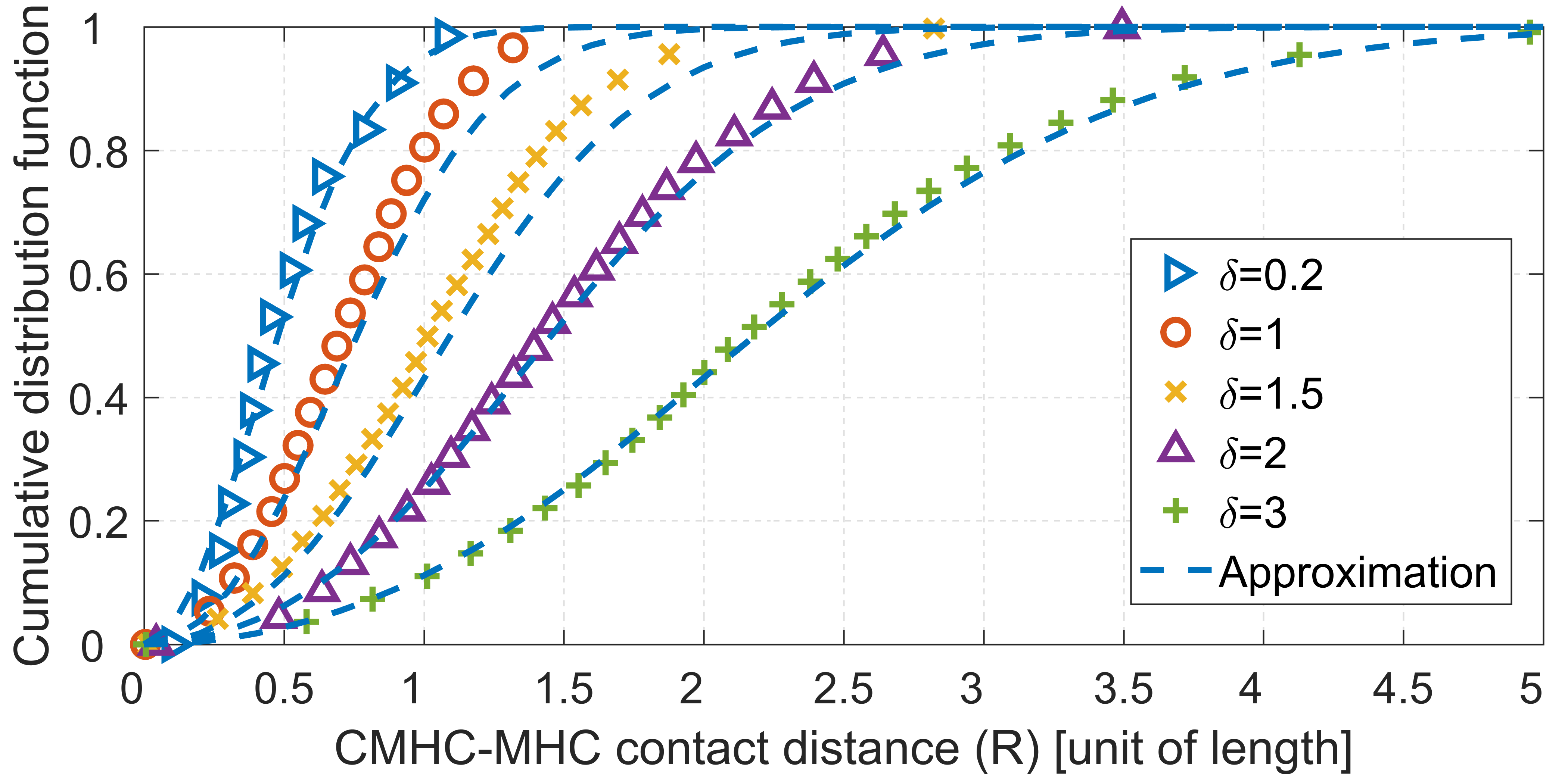}
	\caption{Comparison between simulations and the approximated cumulative distribution function of the distance  between CMHC to its sibling MHC.}\label{Fig_CMHC_MHC}
\end{figure}

\section{Conclusion}
In this paper we presented an analytic framework for obtaining the statistics of the contact distance in hard-core point processes. We applied this framework on the well-known Mat\'{e}rn point process to obtain the contact distance for three different cases; (i) an analytic formula for MHC $\to$ MHC and (ii) PPP $\to$ MHC, and an approximated solution for CMHC $\to$ MHC. We tested the framework against Monte-Carlo simulation, showing good match. Future work might include obtaining an exact expression for CMHC $\to$ MHC and the contact distance statistics of the $n^\mathrm{th}$ neighbour.

\appendix
In order to poof the convergence to Riemann integral we start from the 4$^\mathrm{th}$ line in equation (\ref{Eq_4}) we write:
\begin{equation}\label{Eq_A1}
V(R)	= \prod_{n=2}^{N-1} \mathbb{P} \left[B_n|\bigcap_{i=1}^{n-1}B_i\right]\times \mathbb{P} \left[B_o\right]
\end{equation}
Where we define a function $\hat{V}(n,N)$ such that $V(R) = \prod_{n=2}^{N-1}\hat{V}(n,N)$, let this function be written as 
\begin{equation}\label{Eq_A3}
\hat{V}(n,N)= 1-\zeta(n,N) \Delta r,
\end{equation}
we substitute in (\ref{Eq_A1}) and take the natural logarithm of both sides, accordingly
\begin{equation}\label{Eq_A4}
\ln\left[\hat{V}(n,N)\right]= \sum_{n=2}^{N-1} \ln\left[1-\zeta(n,N) \Delta r\right],
\end{equation}
where (\ref{Eq_A4}) is valid for all integer choice of $N$, accordingly we take the following limit
\begin{align} \label{Eq_A5}
\ln\left[V(n,N)\right] &= \lim\limits_{N\to\infty}\sum_{n=2}^{N-1} \ln\left[1-\zeta(n,N) \Delta r\right] \nonumber \\
&= \sum_{n=2}^{N-1} \lim\limits_{N\to\infty} \ln\left(\left[1-\zeta(n,N) \right]\Delta r\right)  \nonumber \\
&\stackrel{(a)}{\approx} - \sum_{n=2}^{N-1} \lim\limits_{N\to\infty} \left[\zeta(n,N) \Delta r\right]  \nonumber \\
&=- \lim\limits_{N\to\infty} \sum_{n=2}^{N-1}  \left[\zeta(n,N) \Delta r\right],
\end{align}
where step $(a)$ follows from Taylor expansion of the natural logarithm around zero, where $\ln(1+x)=x+O(x^2)$, thus $\ln(1+x)\approx x$ as $x \to 0$. Since the function $\zeta(n,N)$ is bounded and monotonic, then the summation (\ref{Eq_A5}) is, by definition, either the lower or the upper Riemann sums \cite{hunter2012introduction}, and the function is Riemann integrable as per theorem 1.21 in \cite{hunter2012introduction}. Accordingly the limit in (\ref{Eq_A5}) converges to Riemann integral as per theorem 1.17 in \cite{hunter2012introduction}, thus
\begin{equation}\label{Eq_A6}
\ln\left[V(n,N)\right]= \int_{0}^{R}\lim\limits_{N\to\infty} \left[\zeta(n,N) \Delta r\right],
\end{equation}
where the above limit can be evaluated starting from the definition in (\ref{Eq_A3}) as taking the limit of both sides
\begin{align}
\lim\limits_{N\to\infty} \hat{V}(n,N) &= \lim\limits_{N\to\infty} \left[1-\zeta(n,N) \Delta r\right] \nonumber \\
v_{m \to c}(r) &= 1-\lim\limits_{N\to\infty} \left[\zeta(n,N) \Delta r\right] \nonumber \\
2 \pi r \lambda_p~\eta_{m \to c}(r) ~  \mathrm{d}r &= \lim\limits_{N\to\infty} \left[\zeta(n,N) \Delta r\right]
\end{align}
hence, substituting in (\ref{Eq_A5}) we get the proof.
\ifCLASSOPTIONcaptionsoff

\newpage
\fi

\bibliography{MHC_Contact}

\end{document}